\documentclass[10pt,comsoc, conference]{IEEEtran} 
\usepackage{times}
\usepackage{graphicx}
\usepackage{amssymb}
\usepackage{url,hyperref}
\usepackage{amsmath}
\usepackage{amsfonts}
\usepackage{amsthm}
\usepackage{mathtools}
\usepackage{physics}
\usepackage{bbm}
\usepackage{xfrac}
\usepackage{amssymb}

\newcommand{\cosimcid}{\mathcal{C}_{\text{id}}^{\text{sim}}}
\newcommand{\cocid}{\mathcal{C}_{\text{id}}}
\newcommand{\simcid}{C_{\text{id}}^{\text{sim}}}
\newcommand{\cid}{C_{\text{id}}}

\newcommand{\alx}{\mathcal{X}}
\newcommand{\aly}{\mathcal{Y}}
\newcommand{\hi}{\mathcal{H}}
\newcommand{\st}{\mathcal{S}}

\newcommand{\co}{\mathcal{C}}

\newcommand{\floor}[1]{\lfloor #1 \rfloor}

\newcounter{thm}
\newtheorem{theorem}[thm]{Theorem}

\newtheorem{definition}[thm]{Definition}
\newtheorem{corollary}[thm]{Corollary}
\newtheorem{lemma}[thm]{Lemma}

\begin{document}
\title{The Simultaneous Identification Capacity of the Classical - Quantum Multiple Access Channel with Stochastic Encoders for Transmission} 
	
\author{Holger Boche, Stephen Diadamo\\
	\normalsize{Lehrstuhl f\"ur Theoretische Informationstechnik}\\
	\normalsize{ Technische Universit\"at M\"unchen}\\
	\normalsize{\small{\{boche, stephen.diadamo\}@tum.de}}
}
\maketitle
\thispagestyle{plain}
\pagestyle{plain}

\begin{abstract}
	In \cite{sim_id}, it is shown that the simultaneous identification capacity region for the discrete, memoryless, classical-quantum multiple access channel is equal to the transmission capacity region for codes using a deterministic encoding scheme. Here, we use transmission codes with a stochastic encoding scheme and show that the analogous capacity theorem holds. Further, by using stochastic encoders, the proof of achievability is highly simplified.
\end{abstract}
	
\section{Introduction}
    The process of message identification is to receive a message and identify if it is a particular message of interest or not, thus obtaining one bit of information. This differs from message transmission, where the receiver attempts to decode the message entirely.  Ahlswede and Dueck introduce message identification theory in \cite{ahlswede_dueck_89} with a key result being that, over a classical channel, by using a stochastic encoding scheme for their identification codes, the size of the identification codes grow doubly-exponentially quickly in block-length. For classical-quantum message identification, L\"ober proves an analogous result  \cite{loeber}. In certain cases, because of the effects of quantum mechanics, identication in the usual scence is not possible. For example, if a receiver is not the final receiver of a communication chain, a measurement on the state will affect the output quantum state, ruining the information content. To overcomes this, L\"ober introduces a model called a simultaneous ID-code which requires a single measurement of the output of the channel which can identify every message at the same time. Boche and Diadamo show that over a discrete, memoryless classical-quantum multiple access channel with transmission codes that use a deterministic encoding scheme, the transmission capacity region is equal to its message identification capacity region  \cite{sim_id}. 
     
    In the current work, we define a transmission code for the classical-quantum multiple access channel which uses a stochastic encoding scheme. This is a more natural approach since the identification codes are also using a stochastic encoding scheme. We define the necessary models for the transmission and identification codes and in the following sections, we prove the main capacity theorem for this report.

\section{Notation and Definitions}
    By $\alx$ or $\aly$ we refer to finite alphabets. The $k$-fold product set of an alphabet $\alx$ is denoted as $\alx^k \coloneqq \alx \times ... \times \alx$. The space of quantum states with respect to a particular Hilbert space $\hi$ is denoted as $\mathcal{S}(\hi)$ and the set of linear operators on $\hi$ is denoted as $\mathcal{L}(\hi)$. Here we only consider finite dimensional, complex Hilbert spaces. The set of probability distributions on another set $\alx$ is denoted as $\mathcal{P}(\alx)$. Quantum channels, usually denoted as $W$, in this report are always completely positive and trace preserving maps.  

    \begin{definition}[CCQ channel]
    	A classical-classical-quantum (CCQ) channel, that is, a channel with two classical senders and one quantum receiver, $\mathbf{W}$ is a family,
    	\begin{align*}
    	\mathbf{W} \coloneqq \{ W^k : \alx^k \times \aly^k \rightarrow \mathcal{S}(\mathcal{H}^{\otimes k})\}_{k\in\mathbb{N}}.
    	\end{align*}
    	When $\mathbf{W}$ is a discrete-memoryless classical-classical-quantum channel (DM-CCQ), we refer to $\mathbf{W}$ simply as $W$.
    \end{definition}

    \begin{definition}[$(k,M,N)$-code]
        For a CCQ channel $\mathbf{W}$, a (randomized) $(k,M,N)$-code for  message transmission is the family $\mathcal{C} \coloneqq (P_m, Q_n,D_{mn})_{m=1,n=1}^{M,N}$ where $P_1,...,P_M \in \mathcal{P}(\alx^k)$, $Q_1,...,Q_N\in\mathcal{P}(\aly^k)$, and $(D_{mn})_{m=1,n=1}^{M,N}\subseteq \mathcal{L}(\hi^{\otimes k})$ forms a POVM.
    	\\\\\noindent
    	For a $(k,M,N)$-code $\co$, we define the error as,
    	\begin{align*}
            & e(\co, W^k) \coloneqq \\ &\max_{\substack{m\in[M] \\ n\in[N]}}  1 -  \sum_{\substack{x^k\in\alx^k \\ y^k\in\aly^k}} P_m(x^k)Q_n(y^k)\tr\left( D_{mn}W^k(x^k,y^k) \right), 
	\end{align*}
    \end{definition}
    
    \begin{definition}[$(k,M,N)$-ID-code]
	For a CCQ channel $\mathbf{W}$, a (randomized) $(k,M,N)$-ID-code for classical message identification is the family $\cocid \coloneqq (P_m, Q_n,I_{mn})_{m=1,n=1}^{M,N}$ where $P_1,...,P_M \in \mathcal{P}(\alx^k)$, $Q_1,...,Q_N\in\mathcal{P}(\aly^k)$, and $(I_{mn})_{m=1,n=1}^{M,N}\subseteq \mathcal{L}(\hi^{\otimes k})$ such that $0 \leq I_{mn} \leq \mathbbm{1}_{\hi^{\otimes k}}$ for all $m\in[M]$ and $n \in [N]$.
	\\\\\noindent
	For a $(k,M,N)$-ID-code $\cocid$, we define two types of errors,
	\begin{align*}
        & e_1(\cocid, W^k) \coloneqq \\ &\max_{\substack{m\in[M] \\ n\in[N]}}  1 -  \sum_{\substack{x^k\in\alx^k \\ y^k\in\aly^k}} P_m(x^k)Q_n(y^k)\tr\left( I_{mn}W^k(x^k,y^k) \right), \\
		&e_2(\cocid, W^k) \coloneqq   \\
		&\max_{\substack{m,m'\in[M], \\ n,n'\in[N] \\ 
		(m,n)\neq (m',n')}}   \sum_{\substack{x^k\in\alx^k \\ y^k\in\aly^k}}P_m(x^k)Q_n(y^k)\tr\left(I_{m'n'}W^k(x^k,y^k) \right).
	\end{align*}		
\end{definition}

\begin{definition}[Simultaneous $(k,M,N)$-ID-code]
	A $(k,M,N)$-ID-code $\cocid \coloneqq (P_m, Q_n, I_{mn})_{m=1,n=1}^{M,N}$ is called simultaneous if for $R,S\in\mathbb{N}$ there exists a POVM $(E_{rs})_{r=1,s=1}^{R,S}$ with subsets $A_1, ..., A_M \subset [R]$ and $B_1, ..., B_N \subset [S]$ such that for each $m \in [M]$ and $n\in[N]$, 
	\begin{align*} 
		I_{mn} = \sum_{i\in A_m} \sum_{j \in B_n} E_{ij}. 
	\end{align*}
\end{definition}

\begin{definition}[Achievable rate pair]
	For a CCQ channel $\mathbf{W}$, we say $(R_1, R_2) \in \mathbb{R}^2$, $R_1, R_2 \geq 0$, is an achievable rate pair if for all $\epsilon, \delta > 0$, there exists a $k_0$ such that for all $k\geq k_0$ there exists a $(k,M,N)$-code $\co$ such that, 
	\begin{align*}
		\begin{aligned}
		\frac{1}{k} \log M \geq R_1 - \delta,  \hspace{3mm} \frac{1}{k}\log N \geq R_2 - \delta, \hspace{3mm} e(\co, W^k) \leq \epsilon
		\end{aligned}
	\end{align*}
	The capacity region for $\mathbf{W}$ is defined as
	\begin{align*}
		C_{\text{st}}(\mathbf{W}) \coloneqq \{(R_1, R_2) \mid (R_1, R_2) \text{ is an achievable rate pair} \}. 
	\end{align*}
\end{definition}

\begin{definition}[Achievable simultaneous ID-rate pair] 
	For a CCQ channel $\mathbf{W}$, we say $(R_1, R_2) \in \mathbb{R}^2$, $R_1, R_2 \geq 0$, is an achievable simultaneous ID-rate pair if for $\epsilon_1, \epsilon_2, \delta > 0$, there exists a $k_0$ such that for all $k \geq k_0$ there is a simultaneous $(k,M,N)$-ID-code $\cosimcid$ with
	\begin{align*}
		\begin{aligned}
		\frac{1}{k} \log\log M \geq R_1 - \delta, & \hspace{3mm} \frac{1}{k}\log\log N \geq R_2 - \delta, \\  e_1(\cosimcid, W^k) \leq \epsilon_1,  &\hspace{3mm} e_2(\cosimcid, W^k) \leq \epsilon_2. 
		\end{aligned}
	\end{align*}
	The simultaneous ID capacity region for a CCQ channel $\mathbf{W}$ is defined as		
	\begin{align*}
		  & \cid^{\text{sim}}(\mathbf{W}) \coloneqq \\ &\{(R_1, R_2) \mid (R_1, R_2) \text{ is an achievable sim. ID-rate pair} \}. 
	\end{align*}
\end{definition}

\section{Capacity Theorem}
\begin{theorem}\label{main_result}
	For a DM-CCQ channel generated by $W:\alx\cross\aly\rightarrow\st(\hi)$, 
	\begin{align*}
		C_{\text{st}}(W) = \cid^{\text{sim}}(W). 
	\end{align*}
\end{theorem}

\section{Proof of Achievability}
\begin{theorem}\label{sim_ach}
	For a CCQ channel $\mathbf{W}$, not necessarily discrete and memoryless, 
	\begin{align*}
		C_{\text{st}}(\mathbf{W}) \subseteq \simcid(\mathbf{W}). 
	\end{align*} 
\end{theorem}

\begin{lemma}{\cite[Lemma 3.1]{loeber}}\label{lem:lober}
	Let $M\in \mathbb{N}$ be a finite number and $\epsilon\in(0,1)$. Let $\lambda > 0$ be such that $\epsilon\log(\frac{1}{\lambda} - 1) > 2$. Then, there are at least $N \geq \frac{1}{M}2^{\lfloor \lambda M \rfloor}$ subsets $A_1,...,A_N \subset [M]$ such that each $A_i$ has cardinality $\lfloor \lambda M \rfloor$. Further, the cardinalities of the pairwise intersections satisfy, $\forall \hspace{1mm} i, j \in [N], i\neq j$, 
	\begin{align*}
		|A_i \cap A_j| \leq \epsilon\lfloor\lambda M\rfloor, \hspace{1cm} 
	\end{align*}		
\end{lemma}

\begin{proof}[Proof of Theorem \ref{sim_ach}]
   Let $\epsilon_1, \epsilon_2, \delta \in (0,1)$ and define $\epsilon\coloneqq \min \{\epsilon_1, \sfrac{\epsilon_2}{4}\}$. Let $(R_1, R_2) \in C_{\text{st}}(\mathbf{W})$ be an achievable transmission rate pair and $\lambda >0$ such that $\epsilon \log(\frac{1}{\lambda} - 1)>2$. Let $\co = (P_i,Q_j, D_{ij})_{i=1,j=1}^{M,N}$ be a $(k,M,N)$-code that achieves $(R_1,R_2)$ such that
   \begin{align*}
      \frac{1}{k}\log M \geq R_1 - \delta,\hspace{3mm}\frac{1}{k}\log N \geq R_2 - \delta,\hspace{3mm} e(\co, W^k) \leq \epsilon.
   \end{align*}
   By Lemma \ref{lem:lober}, there exist $M'$ subsets $A_1,..., A_{M'} \subset [M]$ and $N'$ subsets $B_1,..., B_{N'} \subseteq [N]$ such that 
   \[
   M' \geq \frac{1}{M}2^{\floor{\lambda M}} \text{ and }  N' \geq \frac{1}{N'}2^{\floor{\lambda N}}, 
   \] 
   and for $i,j\in[M'], i\neq j$, and $n,m\in[N'], n\neq m$, 
   \begin{align*}
       |A_i| = \floor{\lambda M}, \hspace{2mm} |A_i \cap A_j| \leq \epsilon \floor{\lambda M}, \\ |B_n| = \floor{\lambda N}, \hspace{2mm} |B_n \cap B_m| \leq \epsilon \floor{\lambda N}.
   \end{align*}
   For each $m \in [M']$ define 
   \[ P'_m(x^k) \coloneqq \frac{1}{|A_m|}\sum_{i \in A_m}P_i(x^k),\]
   and each $n \in [N']$ define 
   \[ Q'_n(y^k) \coloneqq \frac{1}{|B_n|}\sum_{j \in B_n}Q_j(y^k).\]
   Further, define
   \[ I_{mn}\coloneqq \sum_{i\in A_m}\sum_{j\in B_n} D_{ij}. \]
   We analyze the two types of errors for the code \[ \cocid\coloneqq(P'_m, Q'_n, I_{mn})_{m=1,n=1}^{M', N'}.\] Fix $m\in[M']$ and $n \in [N']$, then
   \begin{align*}
       1 &- \sum_{x^k, y^k} P'_m(x^k) Q'_n(y^k)\tr( I_{mn}W^k(x^k,y^k) )    \\
       & \begin{aligned} 
        &= 1 - \\ &\hspace{10mm}\frac{1}{|A_m||B_n|} \sum_{\substack{x^k, y^k \\ i\in A_m\\ j\in B_n}} P_i(x^k)Q_j(y^k) \hspace{2mm}\cdot \\   & \hspace{30mm}\tr( \left(\sum_{\substack{a\in A_m\\ b\in B_n}} D_{ab}\right) W^k(x^k,y^k) ) 
       \end{aligned}\\
       &\begin{aligned} 
       \leq  1 - \\ &\hspace{-4mm}\frac{1}{|A_m||B_n|}\sum_{\substack{x^k, y^k \\ i\in A_m\\ j\in B_n}} P_i(x^k)Q_j(y^k)\tr( D_{ij}  W^k(x^k,y^k) ) 
       \end{aligned} \\
       &\leq 1 - \max_{\substack{i \in A_m \\ j \in B_n}}\sum_{x^k, y^k} P_i(x^k)Q_j(y^k)\tr( \left(D_{ij}\right) W^k(x^k,y^k) )\\
       &\leq \epsilon \leq \epsilon_1.
   \end{align*}
   For the second kind error, with fixed $(m,n)\neq (a,b)$, with $m, a \in [M']$ and  $n, b\in[N']$,
   \begin{align*}
       \sum_{x^k, y^k}& P'_m(x^k)Q'_n(y^k)\tr( I_{ab}W^k(x^k,y^k) ) \\
       &= \frac{1}{|A_m||B_n|}\sum_{\substack{x^k, y^k \\ i\in A_m\\ j \in B_n}}P_i(x^k)Q_j(y^k) \tr(  I_{ab}  W^k(x^k,y^k) )\\
       &\begin{aligned} = \frac{1}{|A_m||B_n|}\sum_{x^k, y^k}\left(  \sum_{\substack{ i \in A_m\cap A_a \\ j \in B_n\cap B_b }}\hspace{-2mm}\cdot  + \sum_{\substack{ i \in A_m \setminus A_a \\ j \in B_n \setminus B_b }}\hspace{-2mm}\cdot \hspace{1mm}+ \right.  \\  \left .  \sum_{\substack{ i \in A_m \cap A_a \\ j \in B_n \setminus B_b }}\hspace{-2mm}\cdot +\sum_{\substack{ i \in A_m \setminus A_a \\ j \in B_n \cap B_b }}\hspace{-2mm}\cdot \hspace{2mm} \right),
       \end{aligned}\\
   \end{align*}
   where we denote with $\cdot$ that the term in the sum is the same for each. We examine each piece of the sum individually. For the first part,
   \begin{align*}
      \sum_{x^k, y^k}& \sum_{\substack{ i \in A_m\cap A_a \\ j \in B_n \cap B_b }}P_i(x^k)Q_j(y^k)\tr(  I_{ab}  W^k(x^k,y^k) ) \\
      &=   \sum_{\substack{ i \in A_m\cap A_a \\ j \in B_n\cap B_b }}   \underbrace{ \sum_{x^k, y^k} P_i(x^k)Q_j(y^k)\tr(  I_{ab}  W^k(x^k,y^k) ) }_{\leq 1}\\
      &\leq |A_m \cap A_a| \cdot |B_n \cap B_b| \\
      &\leq \epsilon |A_m||B_n|,
   \end{align*}
   where in the last step we use that $i\neq a$ or $n\neq b$, and so at least one intersection of sets is not between the same set. For the next part of the sum, 
   \begin{align*}
       &\sum_{x^k, y^k}  \sum_{\substack{ i \in A_m \setminus A_a \\ j \in B_n \setminus B_b }}  P_i(x^k)Q_j(y^k)\tr(  I_{ab}  W^k(x^k,y^k) ) \\
       &=  \sum_{\substack{ x^k, y^k \\ i \in A_m \setminus A_a \\ j \in B_n \setminus B_b }}  \hspace{-2mm}P_i(x^k)Q_j(y^k)\tr(  \left( \sum_{\substack{r\in A_a \\ s \in B_b}} D_{rs}  \right)  W^k(x^k,y^k) ) \\
       &\leq \sum_{\substack{ i \in A_m \setminus A_a \\ j \in B_n \setminus B_b }} \underbrace{\sum_{x^k, y^k} P_i(x^k)Q_j(y^k)\tr(  \left( \mathbbm{1} - D_{ij}  \right)  W^k(x^k,y^k) ) }_{\leq \epsilon} \\
       &\leq \epsilon |A_m| |B_n|.
   \end{align*}
   Next,
   \begin{align*}
        & \sum_{x^k, y^k}  \sum_{\substack{ i \in A_m \cap A_a \\ j \in B_n \setminus B_b }}  P_i(x^k)Q_j(y^k)\tr(  I_{ab}  W^k(x^k,y^k) ) \\
        &\leq \sum_{\substack{ i \in A_m \cap A_a \\ j \in B_n \setminus B_b }} \underbrace{\sum_{x^k, y^k} P_i(x^k)Q_j(y^k)\tr(  \left( \mathbbm{1} - D_{ij}  \right)  W^k(x^k,y^k) ) }_{\leq \epsilon} \\
       &\leq \epsilon |A_m| |B_n|.
   \end{align*}
   Finally, 
    \begin{align*}
        &\sum_{x^k, y^k} \sum_{\substack{ i \in A_m \setminus A_a \\ j \in B_n \cap B_b }}  P_i(x^k)Q_j(y^k)\tr(  I_{ab}  W^k(x^k,y^k) ) \\
       &\leq \sum_{\substack{ i \in A_m \setminus A_a \\ j \in B_n \cap B_b }} \underbrace{\sum_{x^k, y^k} P_i(x^k)Q_j(y^k)\tr(  \left( \mathbbm{1} - D_{ij}  \right)  W^k(x^k,y^k) ) }_{\leq \epsilon} \\
       &\leq \epsilon |A_m| |B_n|.
   \end{align*}
   Combining these four sums, it holds that
   \begin{align*}
       &\begin{aligned} \frac{1}{|A_m||B_n|}\sum_{x^k, y^k}\left(  \sum_{\substack{ i \in A_m\cap A_a \\ j \in B_n\cap B_b }}\cdot  + \sum_{\substack{ i \in A_m \setminus A_a \\ j \in B_n \setminus B_b }}  \cdot \hspace{2mm}+ \right.  \\  \left .  \sum_{\substack{ i \in A_m \cap A_a \\ j \in B_n \setminus B_b }}  \cdot +\sum_{\substack{ i \in A_m \setminus A_a \\ j \in B_n \cap B_b }}  \cdot \right)
       \end{aligned}\\
       &\leq \frac{1}{|A_m||B_n|}(4 \epsilon |A_m||B_n|) \\
       &= 4\epsilon \leq \epsilon_2.
   \end{align*}
   Therefore,  with
   \begin{align*}
       M' \geq  \frac{1}{M}2^{\floor{\lambda M}} \geq 2^{\lambda 2^{k(R_1 - \delta)} - k},
   \end{align*}
    and 
   \begin{align*}
       N' \geq  \frac{1}{N}2^{\floor{\lambda N}} \geq 2^{\lambda 2^{k(R_2 - \delta)} - k},
   \end{align*}
    in the limit of $k\rightarrow \infty$, $\cocid$ achieves the simultaneous rate pair $(R_1, R_2)$, and so $(R_1, R_2) \in C_{\text{id}}^{\text{sim}}(\mathbf{W})$.
\end{proof}

\section{Proof of the Converse}
    It is shown in \cite[Theorem 4.7]{sim_id} that $C_{\text{id}}^{\text{sim}}(W)$, $W$ a DM-CCQ channel, is equal to the capacity region $C(W)$ for codes that use deterministic encoders under average error figure of merit. We need simply to show that $C(W) \subseteq C_{\text{st}}(W)$, and the converse follows.
   
    \begin{theorem}\label{thm1}
	For a CCQ channel $\mathbf{W}$, not necessarily discrete and memoryless, 
	\begin{align*}
		 C(\mathbf{W}) \subseteq  C_{\text{st}}(\mathbf{W}).
	\end{align*} 
    \end{theorem}
    
    \begin{proof}
        Let $\epsilon, \delta \in (0,1)$ and $(R_1,R_2) \in C(\mathbf{W})$ be an achievable rate pair in the sense of \cite[Definition 2.6]{sim_id} such that $\co = (x_m, y_n, D_{mn})_{m=1,n=1}^{M,N}$ is a $(k,M,N)$-code with
        \begin{align*}
            \frac{1}{k}\log M \geq R_1 - \delta,& \hspace{2mm}\frac{1}{k}\log N \geq R_2 - \delta, \hspace{2mm} \bar{e}(\co, W^k) \leq \epsilon,
        \end{align*}
       $\bar{e}(\co, W^k)$ the average error. Let $\mathbf{X}\coloneqq \{x_m\}_{m=1}^{M}$ and $\mathbf{Y}\coloneqq \{y_n\}_{n=1}^N$. For $m \in [M], n \in [N]$ define 
        \[ P_m(x^k) \coloneqq \frac{1}{M}\mathbbm{1}_\mathbf{X}(x^k) \text{ and }  Q_n(y^k)\coloneqq \frac{1}{N}\mathbbm{1}_\mathbf{Y}(y^k) \] and form the code \[ \co'\coloneqq (P_m, Q_n, D_{mn})_{m=1,n=1}^{M,N}. \]
        We analyze the error for such a code. Let $m\in[M]$ and $n\in [N]$ be fixed, then it holds
        \begin{align*}
            1 &- \sum_{x^k,y^k}P_m(x^k)Q_n(y^k)\tr( D_{mn}W^k(x^k, y^k))\\ 
            &= 1 -  \frac{1}{MN}\sum_{\substack{m\in[M] \\ n \in [N]}}\tr( D_{mn}W^k(x_m, y_n)) \\
            &= \bar{e}(\co, W^k)\\
            &\leq \epsilon, 
        \end{align*}
        which holds for any $m$ and $n$, so the error $e(\co',W^k) \leq \epsilon$.
    \end{proof}
    
    \begin{corollary}
    For a DM-CCQ channel generated by $W: \alx\cross\aly\rightarrow\st(\hi)$, 
	\begin{align*}
		 C_{\text{id}}^{\text{sim}}(W) \subseteq  C_{\text{st}}(W).
	\end{align*} 
    \end{corollary}
    
    \begin{proof}
       For a DM-CCQ channel generated by $W: \alx\cross\aly\rightarrow\st(\hi)$, it holds that $ C_{\text{id}}^{\text{sim}}(W) = C(W)$. By Theorem \ref{thm1}, $C_{\text{id}}^{\text{sim}}(W) = C(W) \subseteq C_{\text{st}}(W)$.
    \end{proof}


\end{document}